\documentclass[a4paper,11pt]{amsart}
\usepackage[active]{srcltx}

\usepackage{graphicx}
\usepackage{amsmath}
\usepackage{amssymb}
\usepackage{hyperref}
\usepackage{bbm}

\newtheorem{thm}{Theorem}
\newtheorem{lem}[thm]{Lemma}%
\newtheorem{prop}[thm]{Proposition}%
\theoremstyle{remark}
\theoremstyle{plain}

\numberwithin{equation}{section}

\def\NN{{\mathbb N}}
\def\PP{{\mathbb P}}
\def\QQ{{\mathbb Q}}
\def\RR{{\mathbb R}}

\def\ZZ{{\mathbb Z}}

\def\one{{\mathbbm{1}}}

\def\vecb{{\text{\boldmath$b$}}}

\def\vece{{\text{\boldmath$e$}}}

\def\vecq{{\text{\boldmath$q$}}}

\def\vecs{{\text{\boldmath$s$}}}

\def\vecv{{\text{\boldmath$v$}}}

\def\vecw{{\text{\boldmath$w$}}}
\def\vecx{{\text{\boldmath$x$}}}

\def\vecy{{\text{\boldmath$y$}}}
\def\vecz{{\text{\boldmath$z$}}}

\def\vecalf{{\text{\boldmath$\alpha$}}}
\def\vecbeta{{\text{\boldmath$\beta$}}}

\def\vecomega{{\text{\boldmath$\omega$}}}

\def\UB{{\scrB_1^{d-1}}}
\def\US{{\S_1^{d-1}}}

\def\vecnull{{\text{\boldmath$0$}}}

\def\scrA{{\mathcal A}}
\def\scrB{{\mathcal B}}

\def\scrD{{\mathcal D}}

\def\scrG{{\mathcal G}}
\def\scrH{{\mathcal H}}

\def\scrK{{\mathcal K}}
\def\scrL{{\mathcal L}}

\def\scrP{{\mathcal P}}

\def\scrS{{\mathcal S}}

\def\fU{{\mathfrak U}}

\def\gap{\operatorname{gap}}

\def\e{\mathrm{e}}

\def\GL{\operatorname{GL}}

\def\S{\operatorname{S{}}}

\def\SL{\operatorname{SL}}
\def\ASL{\operatorname{ASL}}

\def\SO{\operatorname{SO}}

\def\O{\operatorname{O{}}}
\def\T{\operatorname{T{}}}

\def\vol{\operatorname{vol}}

\def\Area{\operatorname{Area}}

\def\GamG{\Gamma\backslash G}

\def\ASLZ{\ASL(d,\ZZ)}
\def\ASLR{\ASL(d,\RR)}

\def\SLZ{\SL(d,\ZZ)}
\def\SLR{\SL(d,\RR)}

\def\trans{\,^\mathrm{t}\!}

\def\sigmabar{\overline{\sigma}}

\def\eps{\epsilon}

\def\Onder#1#2#3#4#5{#1 \setbox0=\hbox{$#1$}\setbox1=\hbox{$#2$}
       \dimen0=.5\wd0 \dimen1=\dimen0 \dimen2=\dp0 \dimen3=\dimen2
       \advance\dimen0 by .5\wd1 \advance\dimen0 by -#4
       \advance\dimen1 by -.5\wd1 \advance\dimen1 by -#4
       \advance\dimen2 by -#3 \advance\dimen2 by \ht1
       \advance\dimen2 by 0.3ex \advance\dimen3 by #5
        \kern-\dimen0\raisebox{-\dimen2}[0ex][\dimen3]{\box1}
       \kern\dimen1}

\newcommand{\wD}{\widetilde{\scrD}}

\newcommand{\Q}{\mathbb{Q}}
\newcommand{\R}{\mathbb{R}}
\newcommand{\Z}{\mathbb{Z}}

\newcommand{\HS}{{{\S'_1}^{d-1}}}

\newcommand{\sfrac}[2]{{\textstyle \frac {#1}{#2}}}
\newcommand{\col}{\: : \:}

\newcommand{\bn}{\mathbf{0}}

\title[Particle transport in polycrystals]{Generalized linear Boltzmann equations for particle transport in polycrystals}
\author{Jens Marklof}
\author{Andreas Str\"ombergsson}
\address{School of Mathematics, University of Bristol,
Bristol BS8 1TW, U.K.\newline
\rule[0ex]{0ex}{0ex} \hspace{8pt}{\tt j.marklof@bristol.ac.uk}}
\address{Department of Mathematics, Box 480, Uppsala University,
SE-75106 Uppsala, Sweden\newline
\rule[0ex]{0ex}{0ex} \hspace{8pt}{\tt astrombe@math.uu.se}}
\date{\today}
\thanks{The research leading to these results has received funding from the European Research Council under the European Union's Seventh Framework Programme (FP/2007-2013) / ERC Grant Agreement n. 291147.  
J.M.\ thanks the Isaac Newton Institute, Cambridge for its support and hospitality during the semester ``Periodic and Ergodic Spectral Problems.''
A.S.\ is supported by a grant from the G\"oran Gustafsson Foundation for
Research in Natural Sciences and Medicine, and also by the Swedish Research Council Grant 621-2011-3629.}

\begin{document}


\maketitle

\noindent
The linear Boltzmann equation describes the macroscopic transport of a gas of non-interacting point particles in low-density matter. It has wide-ranging applications, including neutron transport, radiative transfer, semiconductors and ocean wave scattering. Recent research shows that the equation fails in highly-correlated media, where the distribution of free path lengths is non-exponential. We investigate this phenomenon in the case of polycrystals whose typical grain size is comparable to the mean free path length. Our principal result is a new generalized linear Boltzmann equation that captures the long-range memory effects in this setting. A key feature is that the distribution of free path lengths has an exponential decay rate, as opposed to a power-law distribution observed in a single crystal.

\section{Introduction}\label{sec:intro}

The Lorentz gas, introduced by Lorentz in the early 1900s \cite{Lorentz05} to model electron transport in metals, has become one of the most prominent objects in non-equilibrium statistical mechanics. It describes a gas of non-interacting point particles in an infinite array of identical spherical scatterers. Lorentz showed, by adapting Boltzmann's classical heuristics for the hard sphere gas, that in the limit of low scatterer density (Boltzmann-Grad limit) the evolution of a macroscopic particle cloud is described by the linear Boltzmann equation. A rigorous derivation of the linear Boltzmann equation from the microscopic dynamics has been given in the seminal papers by Gallavotti  \cite{Gallavotti69}, Spohn  \cite{Spohn78} and Boldrighini, Bunimovich and Sinai \cite{Boldrighini83}, under the assumption that the scatterer configuration is sufficiently disordered. 

For scatterer configurations with long-range correlations, such as crystals or quasicrystals, the linear Boltzmann equation fails and must be replaced by a more general transport equation that takes into account additional memory effects. The failure of the linear Boltzmann equation was first pointed out by Golse \cite{GolseToulouse2008} for periodic scatterer configurations. We subsequently provided a complete microscopic derivation of the correct transport equation in this setting \cite{partII}. An important characteristic is here that the distribution of free path lengths has a power-law tail with diverging second moment \cite{Bourgain98,Caglioti03,Boca07,partIV} and the long-time limit of the transport problem is superdiffusive \cite{super}. Surveys of these and other recent advances on the microscopic justification of generalized Boltzmann equations can be found in \cite{GolseToulouse2008,icmp,ICM2014}. 

Independent of these developments, Larsen has recently proposed a stationary generalized linear Boltzmann equation for homogeneous media with non-exponential path length distribution, non-elastic scattering and additional source terms. We refer the reader to Larsen and Vasques \cite{Larsen11,Vasques14a,Vasques14b} and Frank and Goudon \cite{Frank10} for more details and applications. Prompted by a question of Larsen, the present paper aims to generalize our findings for  periodic scatterer configurations \cite{partIII,partI,partII,partIV} to polycrystals. 

The approach of this study combines the methods of \cite{partI,partII} with equidistribution theorems from \cite{union}, which were originally developed to analyse unions of incommensurable lattices. We will argue that, in the Boltzmann-Grad limit, the time evolution of a particle cloud is governed by the transport equation  
\begin{equation}\label{glB220}
	\big[ \partial_t + \vecv\cdot\nabla_\vecx - \partial_\xi \big] f_t(\vecx,\vecv,\xi,\vecv_+) \\
	= \int_{\S_1^{d-1}}  f_t(\vecx,\vecv_0,0,\vecv) \,
p_\vecnull(\vecv_0,\vecx,\vecv,\xi,\vecv_+) \,
d\vecv_0 
\end{equation}
subject to the initial condition
\begin{equation}\label{ini0}
	\lim_{t\to 0}f_t(\vecx,\vecv,\xi,\vecv_+) = f_0(\vecx,\vecv)\, p(\vecx,\vecv,\xi,\vecv_+) ,
\end{equation}
where $f_0(\vecx,\vecv)$ is the particle density in phase space at time $t=0$ and $p(\vecx,\vecv,\xi,\vecv_+)$ is a stationary solution of \eqref{glB220}. The variables $\xi$ and $\vecv_+$ represent the distance to the next collision and the velocity thereafter. By adapting our techniques for single crystals \cite{partI}, we will compute the collision kernel $p_\vecnull(\vecv_0,\vecx,\vecv,\xi,\vecv_+)$ in terms of the corresponding kernel of each individual grain. The kernel yields the conditional probability measure
\begin{equation}
p_\vecnull(\vecv_0,\vecx,\vecv,\xi,\vecv_+) \, d\xi\, d\vecv_+
\end{equation}
for the distribution of $(\xi,\vecv_+)\in\RR_{>0}\times \US$ conditional on $\vecv_0,\vecx,\vecv$.
If the grain diameters are sufficiently small on the scale of the mean free path length, the collision kernel has an explicit representation in terms of elementary functions. This yields particularly simple formulas in dimension $d=2$. The necessity of extending the phase space  has already been observed in the case of a single crystal \cite{Caglioti10,partII}, finite unions \cite{union} and in quasicrystals \cite{quasi,quasikinetic}, where the collision kernel is independent of $\vecx$. If the scatterer configuration is disordered and no long-range correlations are present, the dynamics reduces in the Boltzmann-Grad limit to the classical linear Boltzmann equation \cite{Boldrighini83,Gallavotti69,Spohn78}.

The generalized linear Boltzmann equation \eqref{glB220} can be understood as the Fokker-Planck-Kolmorgorov equation (backward Kolomogorov equation) of the following Markovian random flight process: Consider a test particle travelling with constant speed along the random trajectory
\begin{equation}\label{xt}
\vecx(t) = \vecx_{\nu_t} + (t-T_{\nu_t})\vecv_{\nu_t} , \qquad \vecx(0)=\vecx_0, \qquad
\vecv(t)=\vecv_{\nu_t} , \qquad \vecv(0)=\vecv_0,
\end{equation}
where 
\begin{equation}\label{xnqn}
\vecx_n = \vecx_0+\vecq_n,\qquad \vecq_n = \sum_{j=1}^{n} \vecv_{j-1} \xi_j ,\qquad T_n := \sum_{j=1}^n \xi_j, \qquad T_0:=0,
\end{equation}
are the location, displacement and time of the $n$th collision, $\vecv_n$ the velocity after the $n$th collision, and
\begin{equation}\label{xz}
\nu_t := \max\{ n\in\ZZ_{\geq 0} : T_n \leq t \}
\end{equation}
is the number of collisions within time $t$. The above process is determined by the sequence of random variables $(\xi_j,\vecv_j)_{j\in\NN}$ and $(\vecx_0,\vecv_0)$, where $(\vecx_0,\vecv_0,\xi_1,\vecv_1)$ is distributed according to
\begin{equation}
f_0(\vecx_0,\vecv_0)\, p(\vecx_0,\vecv_0,\xi_1,\vecv_1) \,d\vecx_0\,d\vecv_0\,d\xi_1\,d\vecv_1,
\end{equation}
with $f_0$ now being an arbitrary probability density,
and $(\xi_n,\vecv_n)$ is distributed according to
\begin{equation}
p_\vecnull(\vecv_{n-2},\vecx_{n-1},\vecv_{n-1},\xi_{n},\vecv_{n})\, d\xi_{n}\,d\vecv_{n},
\end{equation}
conditional on $(\xi_j,\vecv_j)_{j=1}^{n-1}$ and $(\vecx_0,\vecv_0)$.

{\em Acknowledgements.} JM would like to thank Martin Frank, Kai Krycki and Edward Larsen for the stimulating discussions during his visit to RWTH Aachen in June 2014, and in particular Edward Larsen for suggesting the problem of transport in polycrystals. We thank Dave Rowenhorst for providing us with the image in Fig.~\ref{fig1}.

\section{The setting}\label{sec:setting}

Let $\{\scrG_i\}_i$ be a countable collection of non-overlapping convex open domains in $\RR^d$, and $\{\scrL_i\}_i$ a collection of affine lattices of covolume one. We can write each such affine lattice as $\scrL_i=(\Z^d+\vecomega_i) M_i$ with row vector $\vecomega_i\in\R^d$ and matrix $M_i\in\SL(d,\R)$. We define a {\em polylattice} $\scrP_\epsilon$ as the point set 
\begin{equation}\label{PS}
\scrP_\epsilon = \bigcup_i \big(\scrG_i \cap \epsilon \scrL_i\big),
\end{equation}
where $\epsilon>0$ is a scaling parameter. We refer to $\scrG_i$ as a {\em grain} of $\scrP_\epsilon$.
An example for $\{\scrG_i\}_i$ is a collection of convex polyhedra that tesselate $\RR^d$, i.e., $\cup_i\overline{\scrG_i}=\RR^d$. In general we will, however, allow gaps between grains. The standing assumption in this paper is that the number of grains intersecting any bounded subset of $\R^d$ is finite. The assumption that grains are convex will allow us to ignore correlations of trajectories that re-enter the same grain without intermediate scattering. It is not difficult to extend the present analysis to include these effects. The assumption that all lattices have the same covolume is made solely to simplify the presentation and can easily be removed. Figure \ref{fig1}, reproduced from \cite{Rowenhorst10}, shows the grains of an actual polycrystal sample, the $\beta$-titanium alloy Ti--21S. The $\beta$-form of titanium has a body-centered cubic lattice, which can be represented as the linear deformation $\ZZ^3 M$ of the cubic lattice $\ZZ^3$, where
\begin{equation}\label{BCC}
M=\begin{pmatrix} 
2^{1/3} & 0 & 0 \\
0 & 2^{1/3} & 0 \\
2^{-2/3} & 2^{-2/3} & 2^{-2/3} 
\end{pmatrix} .
\end{equation}

\begin{figure}
\includegraphics[width=0.8\textwidth]{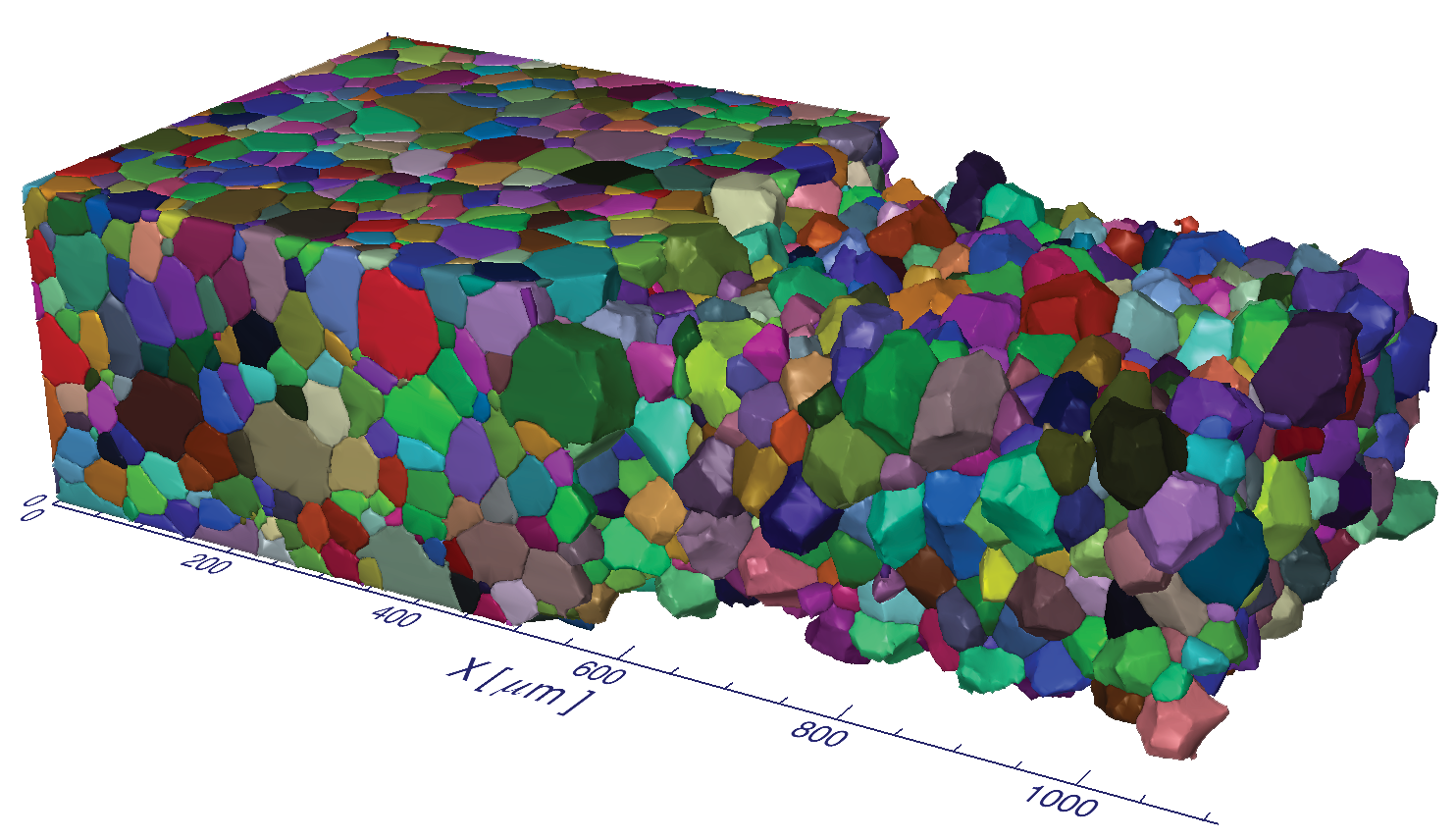}
\caption{Grains in a sample of the $\beta$-titanium alloy Ti--21S. The image is reproduced from ref.~ \cite{Rowenhorst10} by Rowenhorst, Lewis and Spanos.} \label{fig1}
\end{figure}


To model the microscopic dynamics in a polycrystal, we place at each point in $\scrP_\epsilon$ a spherical scatterer of radius $r>0$, and consider a point particle that moves freely until it hits a sphere, where it is scattered, e.g. by elastic reflection (as in the classic setting of the Lorentz gas) or by the force of a spherically symmetric potential. 
We denote the position and velocity at time $t$ by $\vecx(t)$ and $\vecv(t)$. Since (i) the particle speed outside the scatterers is a constant of motion and (ii) the scattering is elastic, we may assume without loss of generality $\|\vecv(t)\|=1$. The dynamics thus takes place in the unit tangent bundle $\T^1(\scrK_{\eps,r})$
where $\scrK_{\eps,r}\subset\RR^d$ is the complement of the set $\scrB^d_r + \scrP_\epsilon$. Here $\scrB^d_r$ denotes the open ball of radius $r$, centered at the origin. We parametrize $\T^1(\scrK_{\eps,r})$ by $(\vecx,\vecv)\in\scrK_{\eps,r}\times\S_1^{d-1}$, where we use the convention that for $\vecx\in\partial\scrK_{\eps,r}$ the vector $\vecv$ points away from the scatterer (so that $\vecv$ describes the velocity {\em after} the collision). 
The Liouville measure on $\T^1(\scrK_{\eps,r})$ is $d\nu(\vecx,\vecv)=d\vecx\,d\vecv$, where $d\vecx=d\!\vol_{\RR^d}(\vecx)$ and $d\vecv=d\!\vol_{\S_1^{d-1}}(\vecv)$ refer to the Lebesgue measures on $\RR^d$ 
and $\S_1^{d-1}$, respectively.

\section{Free path length}\label{sec:Free}

The first collision time with respect to the initial condition $(\vecx,\vecv)\in\T^1(\scrK_{\eps,r})$ is 
\begin{equation} \label{TAU1DEF0}
	\tau_1(\vecx,\vecv) = \inf\{ t>0 : \vecx+t\vecv \notin\scrK_{\eps,r} \}. 
\end{equation}
Since all particles are moving with unit speed, we may also refer to $\tau_1(\vecx,\vecv)$ as the free path length. The mean free path length, i.e.\ the average time between collisions, is for $r\to 0$ asymptotic to $\sigmabar^{\,-1} \epsilon^d r^{-(d-1)}$ where $\sigmabar=\vol\UB$ (the total scattering cross section in units of $r$); this calculation only takes into account the time travelled inside the grains. The scaling limit we are interested in is when the typical grain size is of the order of the mean free path length. We choose (without loss of generality) $\epsilon=r^{(d-1)/d}$ and fix this relation for the rest of this paper. The mean free path length in these units is thus $\sigmabar^{\,-1}$.

Given $(\vecx,\vecv)\in\RR^d\times\US$, we call the sequence $(i_\nu)_{\nu\in\NN}$ the {\em itinary of $(\vecx,\vecv)$} if $i_\nu=i_\nu(\vecx,\vecv)$ is the index of the $\nu$th grain $\scrG_{i_\nu}$ traversed by the trajectory $\{(\vecx+t\vecv,\vecv) : t\geq 0\}$. We denote by $\ell_\nu^-=\ell_\nu^-(\vecx,\vecv),\ell_\nu^+=\ell_\nu^+(\vecx,\vecv)\in[0,\infty]$ the entry resp.\ exit time for $\scrG_{i_\nu}$. If $\vecx\in \scrG_{i_1}$, or if $\vecx\in \partial\scrG_{i_1}$ and $\vecv$ points towards the grain, we set $\ell_1^-=0$. We furthermore define the sejour time for each grain by $\ell_\nu:=\ell_\nu^+-\ell_\nu^-$. Note that, if $\vecx\in\scrG_{i_1}$ and $s$ is suffciently small so that $\vecx+s\vecv\in\scrG_{i_1}$, then
\begin{equation}\label{shift1}
\ell_1^-(\vecx+s\vecv,\vecv) = \ell_1^-(\vecx,\vecv)=0,\qquad
\ell_1^+(\vecx+s\vecv,\vecv) = \ell_1^+(\vecx,\vecv) - s
\end{equation}
and, for all $\nu\geq 2$,
\begin{equation}\label{shift2}
\ell_\nu^\pm(\vecx+s\vecv,\vecv) = \ell_\nu^\pm(\vecx,\vecv) - s .
\end{equation}

We now consider initial data of the form $(\vecx_{\eps,r},\vecv)=(\vecx+\epsilon \vecq+r \vecbeta(\vecv),\vecv)$, where $\vecv\in\S_1^{d-1}$ is random, $\vecx,\vecq\in\RR^d$ are fixed and $\vecbeta:\S_1^{d-1}\to\RR^{d}$ is some fixed continuous function. For $\vecx_{\eps}:=\vecx+\epsilon \vecq\notin\scrP_\epsilon$ (the particle is not $r$-close to a scatterer), the free path length $\tau_1(\vecx_{\eps,r},\vecv)$ is evidently well defined for $r$ sufficiently small. If $\vecx_{\eps}\in\scrP_\epsilon$ (the particle is $r$-close to a scatterer), we assume in the following that $\vecbeta$ is chosen so that the ray $\vecbeta(\vecv)+\R_{\geq 0}\vecv$ lies completely outside the ball $\scrB_1^d$ for all $\vecv\in\S^{d-1}_1$
(thus $r\vecbeta(\vecv)+\R_{\geq 0}\vecv$ lies outside $\scrB_r^d$ for all $r>0$). 

Let $\scrS$ be the commensurator of $\SL(d,\Z)$ in $\SL(d,\R)$. We have
\begin{align*}
\scrS=
\{(\det T)^{-1/d}T\col T\in\GL(d,\Q),\:\det T>0\},
\end{align*}
cf.\ \cite[Thm.\ 2]{borel}, as well as \cite[Sec.\ 7.3]{studenmund}. We say that matrices  $M_1,M_2,\ldots\in\SLR$ are {\em pairwise incommensurable} if $M_i M_j^{-1}\notin\scrS$ for all $i\neq j$. 
The pairwise incommensurability of $M_1,M_2,\ldots$ is equivalent to the fact that the lattices $\scrL_i= (\Z^d+\vecomega_i) M_i$ ($i=1,2,\ldots$) are pairwise incommensurable, 
in the sense that for any $i\neq j$, $c>0$ and $\vecomega\in\R^d$, the intersection
$\scrL_i\cap(c\scrL_j+\vecomega)$ is contained in some affine linear subspace of dimension strictly less than $d$. A natural example in the present setting would be a sequence of matrices $M_i=MK_i$ with $M=1$ (or $M$ as in \eqref{BCC}) and incommensurable rotation matrices $K_i\in\SO(d)$, corresponding to (body-centered) cubic crystal grains with pairwise incommensurable orientation.

The following two theorems comprise our main results for the distribution of free path length. The first theorem deals with generic initial data, the second when the initial position is near or on a scatterer, but still generic with respect to lattices in other grains.

We will in the following use the notation 
\begin{equation}
D_\Phi(\xi)=\int_\xi^\infty \Phi(\eta) d\eta = 1-\int_0^\xi \Phi(\eta) d\eta   
\end{equation}
for the complementary distribution function of the probability density $\Phi$.

In the following we consider lattices $\scrL_i= \epsilon^{-1}\vecx+(\Z^d+\vecomega_i) M_i$ with an additional shift by $\epsilon^{-1}\vecx$. This looks artificial but is necessary for all subsequent statements to hold. [The problem becomes easier if we assume that $\vecomega_i$ are independent random variables uniformly distributed in the torus $\RR^d/\ZZ^d$. In this case it is fine to use $\scrL_i= (\Z^d+\vecomega_i) M_i$. All of the statements below will also hold in this case.]

\begin{thm}\label{freeThm}
Fix $\vecx\in\RR^d$ and, for all $i\in\NN$, let $\scrL_i= \epsilon^{-1}\vecx+(\Z^d+\vecomega_i) M_i$ with $\vecomega_i\in\R^d$, and $M_i\in\SL(d,\R)$ pairwise incommensurable. Fix $\vecq\in\R^d$ so that $\vecomega_i-\vecq M_i^{-1}\notin\QQ^d$ for all $i$. If $(i_\nu)_{\nu\in\NN}$ is the itinary of $(\vecx,\vecv)$, then, for any Borel probability measure 
$\lambda$ on $\S_1^{d-1}$ and any $\xi\geq 0$,
\begin{equation}\label{FPL}
\lim_{r\to 0} \lambda(\{ \vecv\in\S_1^{d-1} \col  \tau_1(\vecx_{\eps,r},\vecv)\geq \xi \})
= \int_{\xi}^\infty \int_{\US} \Psi(\vecx,\vecv,\eta)\,d\lambda(\vecv)\, d\eta  
\end{equation}
with
\begin{equation}\label{limid}
\Psi(\vecx,\vecv,\xi)
= 
\begin{cases}
\bigl(\prod_{\mu=1}^{\nu-1}  D_\Phi(\ell_{\mu})\bigr)\, \Phi(\xi-\ell_{\nu}^-)  & \text{if $\ell_{\nu}^-\leq \xi<\ell_{\nu}^+$} \\
0 & \text{otherwise,}
\end{cases}
\end{equation}
where $\Phi(\xi)$ is the limit probability density of the free path length in the case of a single lattice and for generic inital data, see \cite[Eq.~(1.21)]{partIV}.
\end{thm}

By \cite[Eq.~(1.23)]{partIV} we have,
\begin{equation}\label{smallxi}
\Phi(\xi)=\sigmabar-\frac{\sigmabar^2}{\zeta(d)}\xi+O(\xi^2),
\end{equation}
where $\zeta(d)$ is the Riemann zeta function and the remainder is non-negative. In dimension $d=2$ the error term in fact vanishes identically for $\xi$ sufficiently small;
indeed, for $0< \xi \leq \frac12$ we have \cite[Theorem 2]{Boca07}
\begin{equation}
\Phi(\xi)=2-\frac{24}{\pi^2}\,\xi 
\end{equation}
and hence
\begin{equation}
D_\Phi(\xi)=1-2\xi+\frac{12}{\pi^2}\,\xi^2 .
\end{equation}
In dimension $d=3$ we have \cite[Corollary (1.6)]{partIV} for $0<\xi\leq\frac14$
\begin{equation}
\Phi(\xi)=\pi-\frac{\pi^2}{\zeta(3)}\xi
+\frac{3\pi^2+16}{2\pi\zeta(3)}\xi^2
\end{equation}
and so
\begin{equation}
D_\Phi(\xi)=1-\pi\xi+\frac{\pi^2}{2\zeta(3)}\xi^2
-\frac{3\pi^2+16}{6\pi\zeta(3)}\xi^3 .
\end{equation}
This means that the limit distribution $\Psi(\vecx,\vecv,\xi)$ is completely explicit, if the diameter of each grain is bounded above by $\frac12$ in dimension $d=2$ resp.\ $\frac14$ in dimension $d=3$.

Let us now turn to initial data near a scatterer. Let us fix a map $K:\S_1^{d-1}\to\SO(d)$ such that
$\vecv K(\vecv)=\vece_1$ for all $\vecv\in\S_1^{d-1}$ where $\vece_1:=(1,0,\ldots,0)$.
We assume that $K$ is smooth when restricted to $\S_1^{d-1}$ minus
one point (see \cite[footnote 3, p.~1968]{partI} for an explicit construction). We denote by $\vecx_\perp$ the orthogonal projection of $\vecx\in\RR^d$ onto the hyperplane perpendicular to $\vece_1$.

\begin{thm}\label{freeThm2}
Fix $\vecx\in\scrG_j$ for some $j\in\NN$, and, for all $i\in\NN$, let $\scrL_i= \epsilon^{-1}\vecx+(\Z^d+\vecomega_i) M_i$ with $\vecomega_i\in\R^d$, and $M_i\in\SL(d,\R)$ pairwise incommensurable. Fix $\vecq\in\RR^d$, such that $\vecx_\epsilon=\vecx+\epsilon\vecq\in\epsilon\scrL_j$ and such that $\vecomega_i-\vecq M_i^{-1}\notin\QQ^d$ for all $i\neq j$. If $(i_\nu)_{\nu\in\NN}$ is the itinary of $(\vecx,\vecv)$, then, for any Borel probability measure 
$\lambda$ on $\S_1^{d-1}$ and any $\xi\geq 0$,
\begin{equation}\label{FPL2}
\lim_{r\to 0} \lambda(\{ \vecv\in\S_1^{d-1} \col  \tau_1(\vecx_{\eps,r},\vecv)\geq \xi \})
= \int_\xi^\infty \int_{\US} \Psi_\vecnull(\vecx,\vecv,\eta,(\vecbeta(\vecv) K(\vecv))_\perp)\,d\lambda(\vecv) \,d\eta
\end{equation}
with $\Psi_\vecnull(\vecx,\vecv,\xi,\vecw)$ defined for any $\vecx\in\cup_j \scrG_j$ by
\begin{equation}\label{limid2}
\Psi_\vecnull(\vecx,\vecv,\xi,\vecw)
= 
\begin{cases}
\Phi_\vecnull(\xi,\vecw) & \text{if $0\leq \xi<\ell_{1}^+$} \\
\Phi(\ell_{1},\vecw) \, \bigl(\prod_{\mu=2}^{\nu-1} D_\Phi(\ell_{\mu})\bigr) \, \Phi(\xi-\ell_{\nu}^-)   & \text{if $\ell_{\nu}^-\leq \xi<\ell_{\nu}^+$ ($\nu\geq2$)} \\
0 & \text{otherwise,}
\end{cases}
\end{equation}
where $\Phi_\vecnull(\xi,\vecw)$ is the corresponding limit probability density in the case of a single lattice,
and $\Phi(\xi,\vecw)=\int_\xi^\infty\Phi_\vecnull(\eta,\vecw)\,d\eta$.
\end{thm}

The single-lattice density is given by $\Phi_\bn(\xi,\vecw)=\int_{\scrB_1^{d-1}}\Phi_\bn(\xi,\vecw,\vecz)\,d\vecz$
with $\Phi_\bn(\xi,\vecw,\vecz)$ as in \cite[Sect.~1.1]{partIV}
(cf.\ also Thm.\ \ref{exactpos12} below).
We extend the definition of $\Psi_\vecnull(\vecx,\vecv,\xi,\vecw)$ to all $\vecx\in\RR^d$ as follows. Given a grain $\scrG_j$ and $(\vecx,\vecv)$ with $\vecx\in\partial \scrG_j$, we say $\vecv$ {\em is pointing inwards} if there exists some $\epsilon_0>0$ such that $\{\vecx+\epsilon\vecv \col 0<\epsilon< \epsilon_0\}\subset \scrG_j$. Let
\begin{equation}
\scrH_j:=\{ (\vecx,\vecv) \in\partial\scrG_j\times\US : \text{$\vecv$ is pointing inwards}\}
\end{equation}
and
\begin{equation}
\widehat\scrG_j:= \big(\scrG_j\times\US \big)\cup \scrH_j .
\end{equation}
We now extend the definition of $\Psi_\vecnull(\vecx,\vecv,\xi,\vecw)$ to all $\vecx\in\RR^d$, $\xi>0$, by setting
\begin{equation}
\Psi_\vecnull(\vecx,\vecv,\xi,\vecw) =
\begin{cases}
\lim_{\epsilon\to 0_+} \Psi_\vecnull(\vecx+\epsilon\vecv,\vecv,\xi-\epsilon,\vecw) 
& \text{if $(\vecx,\vecv)\in \cup_j \scrH_j$} \\
0 & \text{if $(\vecx,\vecv)\notin \cup_j \widehat\scrG_j$}. 
\end{cases}
\end{equation}
Let us furthermore define
\begin{equation}
\one(\vecx,\vecv) = 
\begin{cases}
1 & \text{if $(\vecx,\vecv)\in\cup_i \widehat\scrG_i$,}\\
0 & \text{otherwise,}
\end{cases}
\end{equation}
and the differential operator $\scrD$ by (assume $\xi>0$)
\begin{equation}
\scrD \Psi(\vecx,\vecv,\xi) = 
\lim_{\epsilon\to 0_+} \epsilon^{-1}[\Psi(\vecx+\epsilon\vecv,\vecv,\xi-\epsilon)-\Psi(\vecx,\vecv,\xi)]
.
\end{equation}
Note that we have $\scrD \Psi(\vecx,\vecv,\xi)=\big[ \vecv\cdot\nabla_\vecx - \partial_\xi \big] \Psi(\vecx,\vecv,\xi)$ wherever the right-hand side is well defined (which is the case
on a set of full measure).
In the case of a single lattice, we have \cite[Eq.~(1.21)]{partIV}
\begin{equation}\label{sc001}
 \Phi(\xi) = \int_\xi^\infty \int_\UB \Phi_\vecnull(\eta,\vecw)\, d\vecw \,d\eta.
\end{equation}
In the case of a polylattice, \eqref{sc001} generalizes to 
\begin{equation}\label{Cau1}
\begin{cases}
\scrD \Psi(\vecx,\vecv,\xi)= \int_{\UB}  \Psi_\vecnull(\vecx,\vecv,\xi,\vecw)\,d\vecw & (\xi>0) \\
\Psi(\vecx,\vecv,0) = \sigmabar\, \one(\vecx,\vecv).  & 
\end{cases}
\end{equation}
This relation follows from \eqref{limid}, \eqref{limid2} and \eqref{sc001} in view of the relations \eqref{shift1}, \eqref{shift2}. 


\section{The transition kernel}\label{sec:Transition}

To go beyond the distribution of free path length, and towards a full understanding of the particle dynamics in the Boltzmann-Grad limit, we need to refine the results of the previous section and consider the joint distribution of the free path length and the precise location {\em on} the scatterer where the particle hits.

Given initial data $(\vecx,\vecv)$, we denote the position of impact on the first scatterer by
\begin{equation}
	\vecx_1(\vecx,\vecv) := \vecx+\tau_1(\vecx,\vecv) \vecv .
\end{equation}
Given the scatterer location $\vecy\in\scrP_\eps$, we have
$\vecx_1(\vecx,\vecv)\in \S_r^{d-1} + \vecy$ and therefore there is a unique point
$\vecw_1(\vecx,\vecv)\in \S_1^{d-1}$ such that
$\vecx_1(\vecx,\vecv)=r \vecw_1(\vecx,\vecv)+\vecy$. 
It is evident that $-\vecw_1(\vecx,\vecv) K(\vecv)\in \HS$, with the hemisphere $\HS=\{\vecv=(v_1,\ldots,v_d)\in\S_1^{d-1} \col v_1>0\}$. The impact parameter of the first collision is $\vecb=(\vecw_1(\vecx,\vecv) K(\vecv))_\perp$.

As in Section \ref{sec:Free}, we will use the initial data $(\vecx_{\eps,r},\vecv)=(\vecx+\epsilon \vecq+r \vecbeta(\vecv),\vecv)$, where $\vecv\in\S_1^{d-1}$ is random, $\vecx,\vecq\in\RR^d$ are fixed and $\vecbeta:\S_1^{d-1}\to\RR^{d}$ is some fixed continuous function. 

We again have two theorems, the first for generic initial data, the second when the initial position is near or on a scatterer, but still generic with respect to lattices in other grains. Theorems \ref{freeThm} resp.\ \ref{freeThm2} follow from Theorems \ref{exactpos1} resp.\ \ref {exactpos12} below by taking the test set $\fU=\HS$.

\begin{thm}\label{exactpos1}
Fix $\vecx\in\RR^d$ and, for all $i\in\NN$, let $\scrL_i= \epsilon^{-1}\vecx+(\Z^d+\vecomega_i) M_i$ with $\vecomega_i\in\R^d$, and $M_i\in\SL(d,\R)$ pairwise incommensurable. Fix $\vecq\in\R^d$ so that $\vecomega_i-\vecq M_i^{-1}\notin\QQ^d$ for all $i$. If $(i_\nu)_{\nu\in\NN}$ is the itinary of $(\vecx,\vecv)$, then for any Borel probability measure 
$\lambda$ on $\S_1^{d-1}$ absolutely 
continuous with respect to $\vol_{\S_1^{d-1}}$, any
subset $\fU\subset\HS$ with $\vol_{\S_1^{d-1}}(\partial\fU)=0$, 
and any $0\leq a< b$, we have
\begin{multline} \label{exactpos1eq}
\lim_{r\to 0}  \lambda\bigl(\bigl\{ \vecv\in\S_1^{d-1} \col 
\tau_1\in [a,b), \:  
-\vecw_1K(\vecv)\in\fU \bigr\}\bigr) \\
=\int_{a}^{b} \int_{\fU_\perp} \int_{\S_1^{d-1}} 
\Psi\bigl(\vecx,\vecv,\xi,\vecw) 
\, d\lambda(\vecv)\, d\vecw \, d\xi,
\end{multline}
where 
\begin{equation}\label{limid3}
\Psi(\vecx,\vecv,\xi,\vecw)
= 
\begin{cases}
\bigl(\prod_{\mu=1}^{\nu-1}  D_\Phi(\ell_{\mu})\bigr)
\,  \Phi(\xi-\ell_{\nu}^-,\vecw)  & \text{if $\ell_{\nu}^-\leq \xi<\ell_{\nu}^+$} \\
0 & \text{otherwise.}
\end{cases}
\end{equation}
\end{thm}

\vspace{5pt}

\begin{thm}\label{exactpos12}
Fix $\vecx\in\scrG_j$ for some $j\in\NN$, and, for all $i\in\NN$, let $\scrL_i= \epsilon^{-1}\vecx+(\Z^d+\vecomega_i) M_i$ with $\vecomega_i\in\R^d$, and $M_i\in\SL(d,\R)$ pairwise incommensurable. Fix $\vecq\in\RR^d$, such that $\vecx_\epsilon=\vecx+\epsilon\vecq\in\epsilon\scrL_j$ and such that $\vecomega_i-\vecq M_i^{-1}\notin\QQ^d$ for all $i\neq j$. If $(i_\nu)_{\nu\in\NN}$ is the itinary of $(\vecx,\vecv)$, then for any Borel probability measure 
$\lambda$ on $\S_1^{d-1}$ absolutely 
continuous with respect to $\vol_{\S_1^{d-1}}$, any
subset $\fU\subset\HS$ with $\vol_{\S_1^{d-1}}(\partial\fU)=0$, 
and any $0\leq a< b$, we have
\begin{multline} \label{exactpos1eq2}
\lim_{r\to 0}  \lambda\bigl(\bigl\{ \vecv\in\S_1^{d-1} \col 
\tau_1\in [a,b), \:  
-\vecw_1K(\vecv)\in\fU \bigr\}\bigr) \\
=\int_{a}^{b} \int_{\fU_\perp} \int_{\S_1^{d-1}} 
\Psi_\vecnull\bigl(\vecx,\vecv,\xi,\vecw,(\vecbeta(\vecv)K(\vecv))_\perp\bigr) 
\, d\lambda(\vecv)\, d\vecw \, d\xi
\end{multline}
with
\begin{equation}\label{limid22}
\Psi_\vecnull(\vecx,\vecv,\xi,\vecw,\vecz)
= 
\begin{cases}
\Phi_\vecnull(\xi,\vecw,\vecz) & \text{if $0\leq \xi<\ell_{1}^+$} \\
\Phi(\ell_{1},\vecz) \, \bigl(\prod_{\mu=2}^{\nu-1} D_\Phi(\ell_{\mu})\bigr)\, \Phi(\xi-\ell_{\nu}^-,\vecw)    & \text{if $\ell_{\nu}^-\leq \xi<\ell_{\nu}^+$ ($\nu\geq2$),} \\
0 & \text{otherwise,}
\end{cases}
\end{equation}
where $\Phi_\vecnull(\xi,\vecw,\vecz)$ is the transition kernel for a single lattice, cf.~\cite[Sect.~1.1]{partIV}.
\end{thm}

As above, we extend the definition of $\Psi_\vecnull(\vecx,\vecv,\xi,\vecw,\vecz)$ to all $\vecx\in\RR^d$ by setting
\begin{equation}
\Psi_\vecnull(\vecx,\vecv,\xi,\vecw,\vecz) =
\begin{cases}
\lim_{\epsilon\to 0_+} \Psi_\vecnull(\vecx+\epsilon\vecv,\vecv,\xi-\epsilon,\vecw,\vecz) 
& \text{if $(\vecx,\vecv)\in \cup_j \scrH_j$} \\
0 & \text{if $(\vecx,\vecv)\notin \cup_j \widehat\scrG_j$}. 
\end{cases}
\end{equation}

We refer the reader to \cite{partI,partIV} for a detailed study of $\Phi_{\vecnull}(\xi,\vecw,\vecz)$, $\Phi(\xi,\vecw)$ and $\Phi_{\vecnull}(\xi,\vecw)$,
which are related via \cite[Eq.~(6.67)]{partII},
\begin{equation}\label{fit}
\Phi(\xi,\vecw) =\int_{\xi}^\infty \int_{\scrB_1^{d-1}} \Phi_{\vecnull}(\eta,\vecw,\vecz)\, d\vecz\,d\eta 
\end{equation} 
and
\begin{equation}\label{fit22}
\Phi_\vecnull(\xi,\vecw) =\int_{\scrB_1^{d-1}} \Phi_{\vecnull}(\xi,\vecw,\vecz)\, d\vecz .
\end{equation} 
We have in particular \cite[Eq.~(1.18)]{partIV}, 
\begin{align}\label{PHI0ZEROSMALLTHMRES}
\frac{1-2^{d-1}\sigmabar \xi}{\zeta(d)}\leq\Phi_\bn(\xi,\vecw,\vecz)\leq\frac{1}{\zeta(d)} ,
\end{align}
that is, $\Phi_\bn(\xi,\vecw,\vecz)=\zeta(d)^{-1} + O(\xi)$, and \cite[Eq.~(1.19)]{partIV}
\begin{align}
\Phi(\xi,\vecw)
=1-\frac{\sigmabar}{\zeta(d)}\,\xi+O(\xi^2),
\end{align}
where the remainder term is everywhere non-negative, and the implied constant is independent of $\vecw$. As for the free path lengths, we have explicit expressions for these transition kernels in dimensions two and three, which will be discussed in Sections \ref{sec:d2} and \ref{sec:d3}.

The generalization of \eqref{fit} is 
\begin{equation}\label{fitta}
\begin{cases}
\scrD \Psi(\vecx,\vecv,\xi,\vecw)= \int_{\UB}  
\Psi_\vecnull(\vecx,\vecv,\xi,\vecw,\vecz)\,d\vecz & \\
\Psi(\vecx,\vecv,0,\vecw) = \one(\vecx,\vecv) . &
\end{cases}
\end{equation}
Its proof is analogous to \eqref{Cau1}. 

The single-crystal transition kernel satisfies the following invariance properties \cite{partI}:
\begin{equation}
\Phi_{\vecnull}(\xi,\vecz,\vecw) = \Phi_{\vecnull}(\xi,\vecw,\vecz),
\end{equation}
and for all $R\in \O(d-1)$,
\begin{equation}
\Phi_{\vecnull}(\xi,\vecw R,\vecz R) = \Phi_{\vecnull}(\xi,\vecw,\vecz), 
\end{equation}
\begin{equation}
\Phi(\xi,\vecw R) = \Phi(\xi,\vecw), \qquad
\Phi_{\vecnull}(\xi,\vecw R) = \Phi_{\vecnull}(\xi,\vecw).
\end{equation}
These relations imply
\begin{equation}
\Psi_\vecnull(\vecx+\xi\vecv,-\vecv,\xi,\vecz,\vecw) = \Psi_\vecnull(\vecx,\vecv,\xi,\vecw,\vecz),
\end{equation}
and for all $R\in \O(d-1)$,
\begin{equation}
\Psi_\vecnull(\vecx,\vecv,\xi,\vecw R,\vecz R) = \Psi_\vecnull(\vecx,\vecv,\xi,\vecw,\vecz), 
\end{equation}
\begin{equation}
\Psi(\vecx,\vecv,\xi,\vecw R) = \Psi(\vecx,\vecv,\xi,\vecw), \qquad
\Psi_\vecnull(\vecx,\vecv,\xi,\vecw R) = \Psi_\vecnull(\vecx,\vecv,\xi,\vecw).
\end{equation}

\section{Random point processes and the proof of Theorems \ref{freeThm}--\ref{exactpos12}}

We follow the same strategy as in \cite{partI} but use the refined equidistribution theorems for several lattices from \cite{union}. Recall \eqref{PS}, namely
$\scrP_\epsilon = \bigcup_i \big(\scrG_i \cap \epsilon \scrL_i\big)$, with affine lattices $\scrL_i= \epsilon^{-1}\vecx+(\Z^d+\vecomega_i) M_i$. 
Let
\begin{equation}
A_\epsilon = \begin{pmatrix} \epsilon & \vecnull \\ \trans\vecnull & \epsilon^{-1/(d-1)} \one_{d-1} \end{pmatrix} \in\SLR.
\end{equation}
The idea is to consider the sequence of random point processes (with $\epsilon=r^{(d-1)/d}$)
\begin{equation}\label{RPPdef}
\begin{split}
\Theta_\epsilon(\vecx_{\epsilon,r},\vecv) & :=\epsilon^{-1} \big( [\scrP_\epsilon - (\vecx+\epsilon\vecq )]\setminus\{\vecnull\} - r \vecbeta(\vecv) \big) K(\vecv) A_\epsilon
\end{split}
\end{equation}
(where $\vecv$ is distributed according to $\lambda$) and prove convergence, in finite-dimensional distribution, to a random point process as $\epsilon\to 0$. 
Note that the removal of the origin
in \eqref{RPPdef} 
has an effect only when $\vecalf_i:=\vecomega_i-\vecq M_i^{-1}\in\Z^d$ for some $i$;
in fact we have
\begin{align}
\Theta_\epsilon(\vecx_{\epsilon,r},\vecv)
= \bigcup_i \bigg(\epsilon^{-1}(\scrG_i-\vecx_{\epsilon,r}) \cap \big((\Z^d+\vecalf_i\setminus\{\bn\})M_i-\epsilon^{1/(d-1)} \vecbeta(\vecv)\big)\bigg) K(\vecv) A_\epsilon.
\end{align}

Set $G=\ASLR$, $\Gamma=\ASLZ$, and let $\mu$ be the unique $G$-invariant probability measure on $\GamG$.
We let $\Omega$ be the infinite product space $\Omega=\prod_i \GamG$ (one factor for each $\scrG_i$) and let $\omega$ be the corresponding product measure $\prod_i \mu$.
Let us define, for any $(g_i)\in\Omega$ and $\vecv\in\US$,
\begin{equation}
\Theta(\vecx,\vecv,(g_i)):= \bigcup_i\Bigl[\bigl(\bigl((\scrG_i-\vecx)K(\vecv)\cap\R\vece_1\bigr)\times\R^{d-1}\bigr)
\cap\ZZ^d g_i\Bigr].
\end{equation}

Theorems \ref{exactpos1} and \ref{exactpos12} (and thus Theorems \ref{freeThm} and \ref{freeThm2}) follow from the next two theorems by the same steps as in [20, Sections 6 and 9].

\begin{thm}\label{thm:ld1}
Fix $\vecx\in\RR^d$ and, for all $i\in\NN$, let $\scrL_i= \epsilon^{-1}\vecx+(\Z^d+\vecomega_i) M_i$ with $\vecomega_i\in\R^d$, and $M_i\in\SL(d,\R)$ pairwise incommensurable. Fix $\vecq\in\R^d$ so that $\vecomega_i-\vecq M_i^{-1}\notin\QQ^d$. Then for any Borel probability measure 
$\lambda$ on $\S_1^{d-1}$ absolutely 
continuous with respect to $\vol_{\S_1^{d-1}}$, any bounded sets $\scrB_1,\ldots,\scrB_k\subset\RR^d$ 
with boundary of measure zero, and $m_1,\ldots,m_k\in\ZZ_{\geq 0}$,
\begin{multline}
\lim_{\epsilon\to 0} \lambda\big(\big\{  \vecv\in\US : \, \#(\Theta_\epsilon(\vecx_{\epsilon,r},\vecv) \cap\scrB_l )=m_l
\:\: (\forall l=1,\ldots,k) \big\} \big) \\
=\int_{\US}\omega\bigl(\bigl\{(g_i)\in\Omega\col\#(\Theta(\vecx,\vecv,(g_i))\cap\scrB_l)=m_l\:\: (\forall l=1,\ldots,k)
\bigr\}\bigr)\,d\lambda(\vecv).
\end{multline}
\end{thm}

Now set $G_0=\SLR$, $\Gamma_0=\SLZ$, 
and let $\mu_0$ be the unique $G_0$-invariant probability measure on $\Gamma_0\backslash G_0$;
then let $\widetilde\Omega^{(j)}=(\Gamma_0\backslash G_0)\times\prod_{i\neq j}\GamG$
and let $\widetilde\omega$ be the corresponding product measure $\mu_0\times\prod_{i\neq j}\mu$.
Finally let us define, for any $(g_i)\in\widetilde\Omega^{(j)}$ and $\vecv\in\US$,
\begin{align}\notag
\widetilde\Theta^{(j)}(\vecx,\vecv,(g_i)):= 
\Bigl[\bigl(\bigl((\scrG_j-\vecx)K(\vecv)\cap\R\vece_1\bigr)\times\R^{d-1}\bigr)\cap
\big((\ZZ^d\setminus\{\bn\}) g_j-(\vecbeta(\vecv)K(\vecv))_\perp\big)\Bigr]
\hspace{20pt}
\\
\cup \:\: \bigcup_{i\neq j}\Bigl[\bigl(\bigl((\scrG_i-\vecx)K(\vecv)\cap\R\vece_1\bigr)\times\R^{d-1}\bigr)
\cap\ZZ^d g_i\Bigr].
\end{align}

\begin{thm}\label{thm:ld2}
Fix $\vecx\in\scrG_j$ for some $j\in\NN$, and, for all $i\in\NN$, let $\scrL_i= \epsilon^{-1}\vecx+(\Z^d+\vecomega_i) M_i$ with $\vecomega_i\in\R^d$, and $M_i\in\SL(d,\R)$ pairwise incommensurable. Fix $\vecq\in\RR^d$, such that $\vecx_\epsilon=\vecx+\epsilon\vecq\in\epsilon\scrL_j$ and such that $\vecomega_i-\vecq M_i^{-1}\notin\QQ^d$ for all $i\neq j$. Then for any Borel probability measure 
$\lambda$ on $\S_1^{d-1}$ absolutely 
continuous with respect to $\vol_{\S_1^{d-1}}$, any bounded $\scrB_1,\ldots,\scrB_k\subset\RR^d$
with boundary of measure zero,
and $m_1,\ldots,m_k\in\ZZ_{\geq 0}$,
\begin{multline}
\lim_{\epsilon\to 0} \lambda\big(\big\{  \vecv\in\US : \, \, \#(\Theta_\epsilon(\vecx_{\epsilon,r},\vecv) \cap\scrB_l
)=m_l\:\: (\forall l=1,\ldots,k) \big\} \big) \\
=\int_{\US}\widetilde\omega\bigl(\bigl\{(g_i)\in\widetilde\Omega^{(j)}\col\#(\widetilde\Theta^{(j)}(\vecx,\vecv,(g_i))\cap\scrB_l
)=m_l\:\: (\forall l=1,\ldots,k)
\bigr\}\bigr)\,d\lambda(\vecv).
\end{multline} 
\end{thm}

Theorems \ref{thm:ld1} and \ref{thm:ld2} are implied by \cite[Theorem 10]{union} by the same arguments as in \cite[Section 6]{partI}.

\section{Tail estimates}

We will now show that, unlike the case of single crystals, the distribution of free path lengths, as well as the transition kernels, decay exponentially for large $\xi$. This observation relies on the following bound.

\begin{lem}\label{simpli}
For $\xi\geq 0$,
\begin{equation}
D_\Phi(\xi) \leq \max\big( \e^{-\frac{\sigmabar}{2} \xi},  \e^{-\frac{\zeta(d)}{2}} \big). 
\end{equation}
\end{lem}

\begin{proof}
Since 
$1-x\leq e^{-x}$ for $0\leq x <1$, we have
\begin{equation}
D_\Phi(\xi) \leq \e^{-\int_0^\xi \Phi(\eta)\,d\eta} \leq \e^{-\int_0^\xi \max(\sigmabar-\frac{\sigmabar^2}{\zeta(d)}\eta,0) \,d\eta},
\end{equation}
where the second inequality follows from the positivity of the error term in \eqref{smallxi}. 
\end{proof}

By {\em grain diameter} we mean in the following the largest distance between any two points in a single grain. 
We define the {\em gap function},
$\gap(\vecx,\vecv,\xi)$, to be the total length of the trajectory $\{ \vecx+t\vecv : 0\leq t\leq \xi \}$ that is outside $\cup_j\scrG_j$. Note that $\gap(\vecx,\vecv,\xi)\leq \xi$, and furthermore
\begin{equation}
\gap(\vecx+s\vecv,\vecv,\xi) = \gap(\vecx,\vecv,\xi+s)-\gap(\vecx,\vecv,s)
\end{equation}
for all $s\geq 0$.

\begin{prop}\label{prop:gap}
Assume that all grain diameters are uniformly bounded. Then there are constants $C,\gamma>0$ such that for all $\vecx,\vecv,\xi,\vecw,\vecz$
\begin{equation}\label{eq:gap}
\Psi_\vecnull(\vecx,\vecv,\xi,\vecw,\vecz) \leq C \e^{-\gamma (\xi-\gap(\vecx,\vecv,\xi))}.
\end{equation}
The same bound holds for $\Psi(\vecx,\vecv,\xi,\vecw)$, $\Psi_\vecnull(\vecx,\vecv,\xi,\vecw)$ and $\Psi(\vecx,\vecv,\xi)$.
\end{prop}

\begin{proof}
If the grain diameters are bounded above by $\ell>0$, we have $\ell_i\leq \ell$ for all $i$, and hence in view of Lemma \ref{simpli},
\begin{equation}
D_\Phi(\ell_i)\leq \e^{-\gamma\ell_i}
\end{equation}
for all $i$, where $\gamma=\min(\frac{\sigmabar}{2},\frac{\zeta(d)}{2\ell})$. The desired bound now follows from \eqref{limid22}.
\end{proof}

Therefore, if $\gap(\vecx,\vecv,\xi)\leq \delta \xi$ for some $\delta\in[0,1)$, we have exponential decay in \eqref{eq:gap} with rate $\gamma(1-\delta)$.

\section{Explicit formulas for the transition kernel in dimension $d=2$}\label{sec:d2}

In dimension $d=2$ we have the following explicit formula for the transition probability \cite{partIII}:
\begin{equation}\label{Xp}	\Phi_\vecnull(\xi,\vecw,\vecz)=\frac{6}{\pi^2}\Upsilon\Bigl(1+\frac{\xi^{-1}-\max(|\vecw|,|\vecz|)-1}{|\vecw+\vecz|}\Bigr)
\end{equation}
with
\begin{equation}
	\Upsilon(x)=
\begin{cases} 
0 & \text{if }x\leq 0\\
x & \text{if }0<x<1\\
1 & \text{if }1\leq x,
\end{cases}
\end{equation}
The same formula was also found independently by Caglioti and Golse \cite{Caglioti10} and by Bykovskii and Ustinov \cite{Bykovskii09}, using different methods based on continued fractions. In particular, for all $\xi\leq \frac12$,
\begin{equation}
\Phi_\vecnull(\xi,\vecw,\vecz)=\frac{6}{\pi^2}  
\end{equation}
which is thus independent of $\vecw,\vecz$. We have furthermore \cite{partIII}, again for all $\xi\leq \frac12$,
\begin{equation}
\Phi_\vecnull(\xi,\vecw)=\frac{12}{\pi^2} ,\qquad
\Phi(\xi,\vecw)=1-\frac{12}{\pi^2} \xi.  
\end{equation}
Recall that in dimension $d=2$, the value $\frac12$ is precisely the mean free path length.

\section{Explicit formulas for the transition kernel in dimension $d=3$}\label{sec:d3}

\begin{figure}
\includegraphics[width=0.45\textwidth]{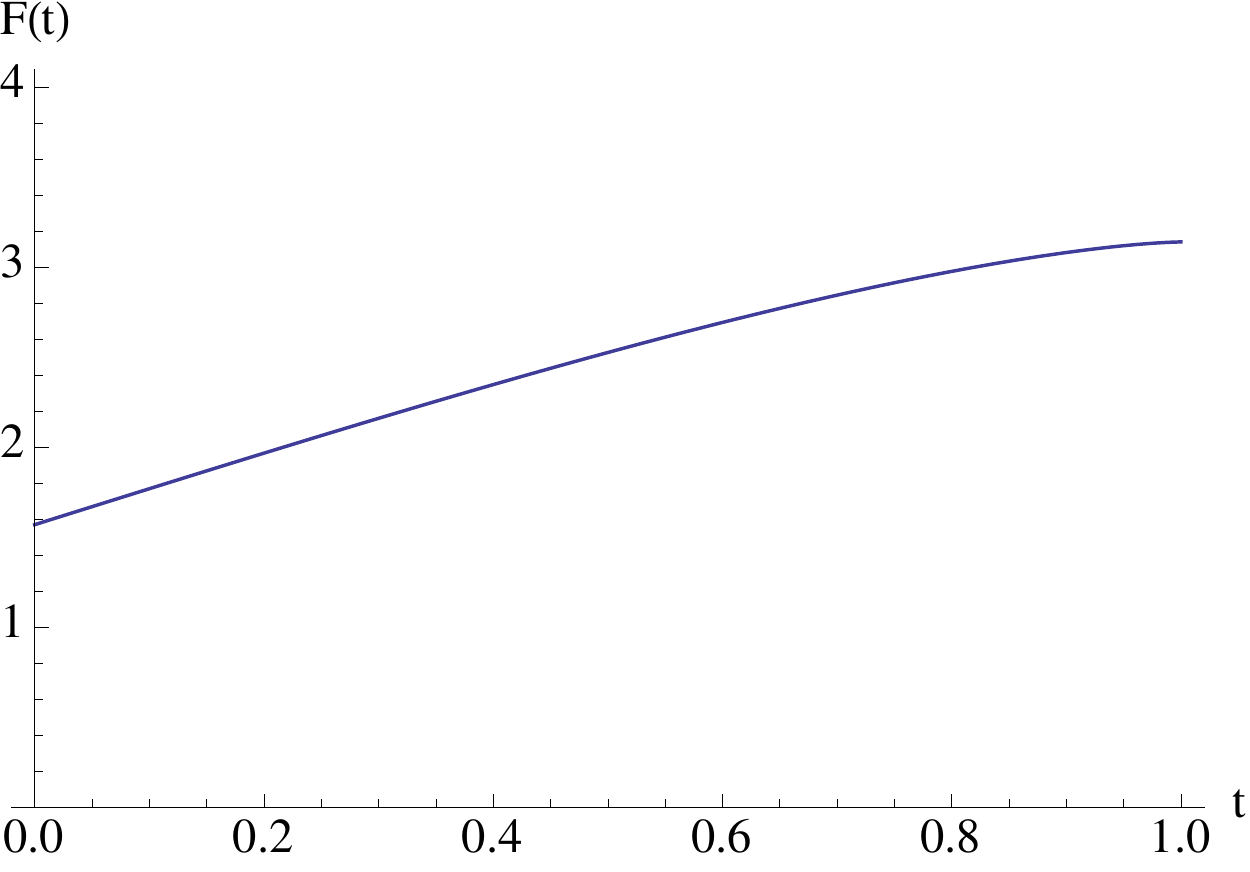}
\includegraphics[width=0.45\textwidth]{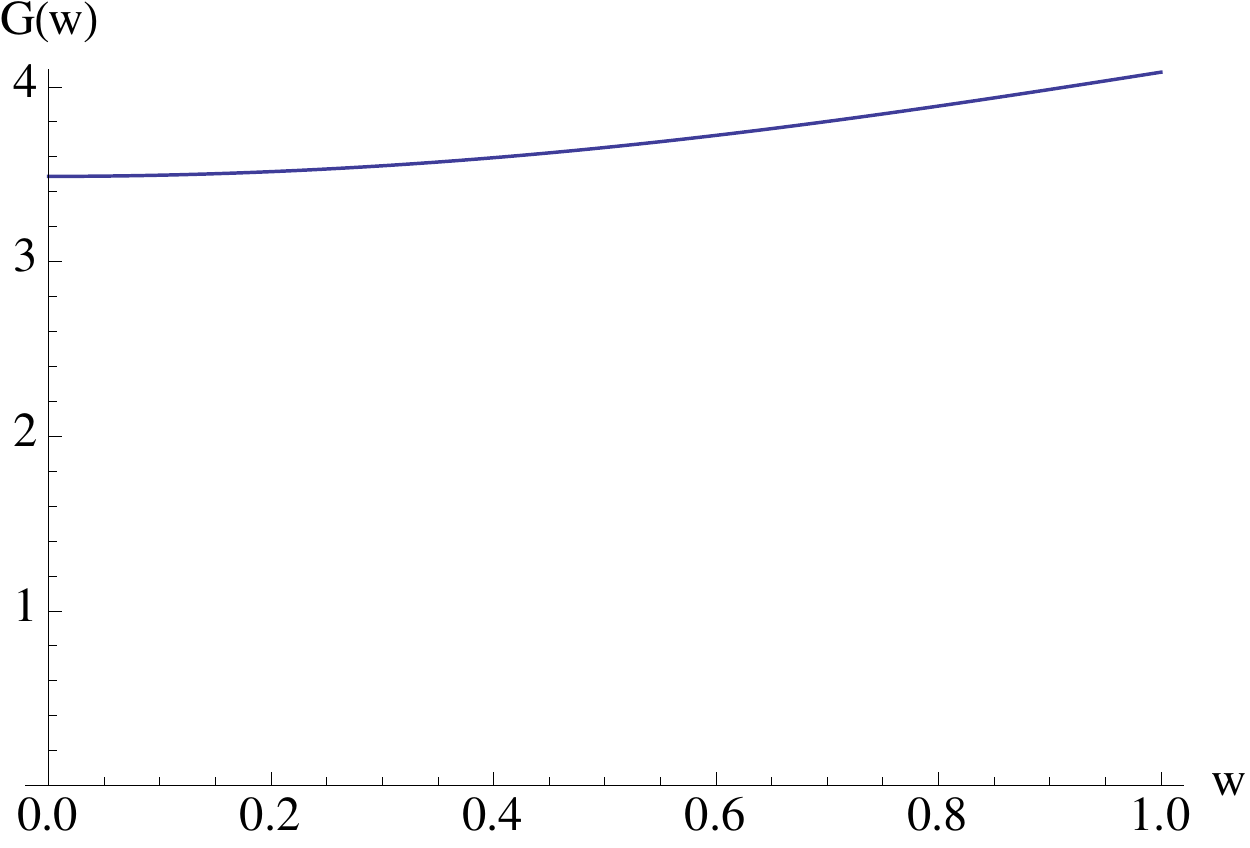}
\caption{The functions $F(t)$ and $G(w)$.} \label{fig2}
\end{figure}

The results in this section are proved in \cite{partIV}. For $0\leq t<1$ set
\begin{equation}\label{FDEF}
F(t)=\pi-\arccos(t)+t\sqrt{1-t^2}
=\Area\bigl(\bigl\{(x_1,x_2)\in\scrB_1^2\col x_1<t\bigr\}\bigr),
\end{equation}
see Fig.~\ref{fig2}.
%
Then, for $0<\xi\leq\frac14$,
\begin{equation}\label{D3EXPLTHMRES}
\Phi_\vecnull(\xi,\vecw,\vecz)=\zeta(3)^{-1}\Bigl(1
-\frac6{\pi^2}F\bigl(\sfrac12\|\vecw-\vecz\|\bigr)\xi\Bigr)
\end{equation}
and
\begin{equation}\label{D3EXPLTHMCOR1RES1}
\Phi(\xi,\vecw)=1-\frac{\pi}{\zeta(3)}\xi
+\frac6{\pi^2\zeta(3)}G(\|\vecw\|)\xi^2,
\end{equation}
where $G:[0,1]\to\R_{>0}$ is the function
\begin{equation}\label{D3EXPLTHMCOR1GDEF}
G(w)=
\pi\int_0^{1-w}F(\sfrac12r)r\,dr
+\int_{1-w}^{1+w}F(\sfrac12r)\arccos\Bigl(\frac{w^2+r^2-1}{2wr}\Bigr)
r\,dr,
\end{equation}
cf.\ Fig.~\ref{fig2}.
The function $G(w)$ is continuous and strictly increasing, and satisfies
$G(0)=\frac{\pi(4\pi+3\sqrt3)}{16}$ and $G(1)=\frac5{16}\pi^2+1$.

\section{The transport equation}\label{sec:GLBE}

In order to prove that the dynamics of a test particle converges, in the Boltzmann-Grad limit $r\to 0$, to a random flight process $(\vecx(t),\vecv(t))$ defined in \eqref{xt}--\eqref{xz}, we require technical refinements of Theorems \ref{exactpos1} and \ref{exactpos12}, where the convergence is uniform over a certain class of $\lambda$. This argument follows the strategy developed in \cite{partII} for single crystals. We will here not attempt to prove these uniform versions, but move straight to the description of the limit process which is determined by the transition kernel of Theorem \ref{exactpos12}.

As in the case of a single crystal \cite{partII}, the limiting random flight process becomes Markovian on an extended phase space, where the additional variables are 
\begin{equation}
\xi(t)=T_{\nu_t+1}-t \in\RR_{>0} \qquad\text{(distance to the next collision)}
\end{equation}
and 
\begin{equation}
\vecv_+(t)=\vecv_{\nu_t+1}\in\US\qquad\text{(velocity after the next collision).}
\end{equation}
The continuous time Markov process $\Xi(t)=(\vecx(t),\vecv(t),\xi(t),\vecv_+(t))$ is determined by the initial distribution $f_0(\vecx,\vecv,\xi,\vecv_+)$ and the collision kernel $p_\vecnull(\vecv_{j-1},\vecx_j,\vecv_j,\xi_{j+1},\vecv_{j+1})$ which yields the probability that the $(j+1)$st collision is at distance $\xi_{j+1}$ from the $j$th collision, with subsequent velocity $\vecv_{j+1}$, given that the $j$th collision takes place at $\vecx_j$ and the particle's velocities before and after this collision are $\vecv_{j-1}$ and $\vecv_j$. The collision kernel $p_\vecnull(\vecv_0,\vecx,\vecv,\xi,\vecv_+)$ is related to the transition kernel $\Psi_\vecnull(\vecx,\vecv,\xi,\vecw,\vecz)$ of the previous sections by
\begin{equation}
	p_\vecnull(\vecv_0,\vecx,\vecv,\xi,\vecv_+) 
	=\Psi_\vecnull(\vecx,\vecv,\xi,\vecb,-\vecs)\,\sigma(\vecv,\vecv_+)
\end{equation}
where $\sigma(\vecv,\vecv_+)$ is the differential cross section, $\vecs=\vecs(\vecv,\vecv_0)$, $\vecb=\vecb(\vecv,\vecv_+)$ are the exit and impact parameters of the previous resp.\ next scattering event.

The density $f_t(\vecx,\vecv,\xi,\vecv_+)$ of the process at time $t>0$ is given by
\begin{equation}
\int_\scrA f_t(\vecx,\vecv,\xi,\vecv_+)\, d\vecx\,d\vecv\,d\xi\,d\vecv_+ = \PP\big(\Xi(t)\in\scrA\big)
\end{equation}
for suitable test sets $\scrA$.
Let us write
\begin{equation}\label{f1}
	f_t(\vecx,\vecv,\xi,\vecv_+) = \sum_{n=0}^\infty f^{(n)}_t(\vecx,\vecv,\xi,\vecv_+)
\end{equation}
where $f^{(n)}_t(\vecx,\vecv,\xi,\vecv_+)$ is the density of particles that have collided precisely $n$ times in the time interval $[0,t]$. Then
\begin{equation}\label{f2}
	f^{(0)}_t(\vecx,\vecv,\xi,\vecv_+) = f_{0}(\vecx-t\vecv,\vecv,\xi+t,\vecv_+),
\end{equation}
and for $n\geq 1$
\begin{multline}\label{f3}
	f^{(n)}_t(\vecx,\vecv,\xi,\vecv_+) = \int_{T_n<t} f_{0}\big(\vecx_0,\vecv_0,\xi_1,\vecv_1\big) \\ \times \prod_{j=1}^n p_\vecnull(\vecv_{j-1},\vecx_j,\vecv_j,\xi_{j+1},\vecv_{j+1}) 	 \, d\xi_n \, d\vecv_{n-1} \cdots d\xi_1\,d\vecv_0  ,
\end{multline} 
with $\vecv_{n+1}=\vecv_+$, $\vecv_n=\vecv$, $\xi_{n+1}=\xi+t-T_n$, $\vecx_0=\vecx-\vecq_n-(t-T_n) \vecv$,
and with $\vecx_j$, $\vecq_n$, $T_n$ as in \eqref{xnqn}.
For general densities $f_{0}(\vecx,\vecv,\xi,\vecv_+)$, relations \eqref{f1}--\eqref{f3} define a family of linear operators (for $t>0$)
\begin{equation}
	K_t^{(n)}f_0(\vecx,\vecv,\xi,\vecv_+):=f^{(n)}_t(\vecx,\vecv,\xi,\vecv_+)
\end{equation}
and
\begin{equation}
	K_t f_0(\vecx,\vecv,\xi,\vecv_+):=f_t(\vecx,\vecv,\xi,\vecv_+).
\end{equation}
One can show that 
\begin{equation}
	\sum_{m=0}^n K_{t_2}^{(n-m)} K_{t_1}^{(m)}=K_{t_1+t_2}^{(n)},
\end{equation}
which in turn implies $K_{t_1+t_2}=K_{t_1} K_{t_2}$, i.e., the operators $K_t$ form a semigroup (reflecting the fact that $\Xi(t)$ is Markovian). The proof of this is analogous to the computation in \cite[Sect.\ 6.2]{partII}.
Hence, for $h>0$ we have $f_{t+h}=K_h f_t$. Since the probability of having more than one collision in a small time interval is negligible, we have for small $h$ (cf.~\cite[Sect.\ 6.2]{partII})
\begin{equation}
	f_{t+h}(\vecx,\vecv,\xi,\vecv_+) = K_h f_t(\vecx,\vecv,\xi,\vecv_+) = K^{(0)}_h f_t(\vecx,\vecv,\xi,\vecv_+) +K^{(1)}_h f_t(\vecx,\vecv,\xi,\vecv_+) +O(h^2).
\end{equation}
Explicitly, we have by \eqref{f2}, \eqref{f3},
\begin{multline}
f_{t+h}(\vecx,\vecv,\xi,\vecv_+) = f_t(\vecx-h\vecv ,\vecv,\xi+h,\vecv_+) \\
+\int_0^h \int_{\US} f_t(\vecx-\xi_1\vecv_0 -(h-\xi_1)\vecv ,\vecv_0,\xi_1,\vecv) \\ \times p_\vecnull(\vecv_0,\vecx-(h-\xi_1)\vecv,\vecv,\xi+h-\xi_1,\vecv_+) \,d\vecv_0 \, d\xi_1 +O(h^2) .   
\end{multline}
Dividing this expression by $h$ and taking the limit $h\to 0$, we obtain 
the Fokker-Planck-Kolmogorov equation (or Kolmogorov backward equation) of the Markov process $\Xi(t)$,
\begin{equation}\label{glB22}
\wD f_t(\vecx,\vecv,\xi,\vecv_+) \\
	= \int_{\S_1^{d-1}}  f_t(\vecx,\vecv_0,0,\vecv) \,
p_\vecnull(\vecv_0,\vecx,\vecv,\xi,\vecv_+) \,
d\vecv_0 ,
\end{equation}
where
\begin{align}
\wD f_t(\vecx,\vecv,\xi,\vecv_+)=
\lim_{\epsilon\to 0_+} \epsilon^{-1}[f_{t+\epsilon}(\vecx+\epsilon\vecv,\vecv,\xi-\epsilon,\vecv_+)-f_t(\vecx,\vecv,\xi,\vecv_+)] .
\end{align}
As for $\scrD$, we observe that $\wD=\partial_t + \vecv\cdot\nabla_\vecx - \partial_\xi$ at any point where the latter operator is well defined (which is the case for a full measure set).
The physically relevant initial condition is
\begin{equation}\label{ini}
	\lim_{t\to 0}f_t(\vecx,\vecv,\xi,\vecv_+) = f_0(\vecx,\vecv,\xi,\vecv_+)= f_0(\vecx,\vecv)\, p(\vecx,\vecv,\xi,\vecv_+) 
\end{equation}
with
	\begin{equation}
	p(\vecx,\vecv,\xi,\vecv_+) := \Psi\big(\vecx,\vecv,\xi,\vecb \big)\, \sigma(\vecv,\vecv_+) .
	\end{equation}
The original phase-space density is recovered via projection,
\begin{equation}
f_t(\vecx,\vecv)=\int_0^\infty \int_{\S_1^{d-1}} f_t(\vecx,\vecv,\xi,\vecv_+)\, d\vecv_+\, d\xi .
\end{equation}
Note that \eqref{fitta} implies that $f_t(\vecx,\vecv,\xi,\vecv_+)=p(\vecx,\vecv,\xi,\vecv_+)$ is the stationary solution of \eqref{glB22}, corresponding to $f_0(\vecx,\vecv)=1$. Uniqueness in the Cauchy problem \eqref{glB22}--\eqref{ini} follows from standard arguments, cf.\ \cite[Section 6.3]{partII}.

The generalized linear Boltzmann equation \eqref{glB22} will also hold for other grainy materials, provided different grains are uncorrelated to guarantee the factorization of the individual grain-distribution functions in \eqref{limid22}. We will discuss the simplest example, grains of a disordered medium, in Section \ref{sec:disorder}. Note that it is not necessary that the transition probabilities $\Phi_\vecnull(\xi,\vecw,\vecz)$ in each grain are identical---the modifications in  \eqref{limid22} are straightforward: replace $\Phi_\vecnull(\xi,\vecw,\vecz)$, $\Phi(\xi,\vecw)$, etc.\ by the grain-dependent $\Phi_\vecnull^{(i_\nu)}(\xi,\vecw,\vecz)$, $\Phi^{(i_\nu)}(\xi,\vecw)$, etc.\ throughout. 

\section{A simplified kernel}\label{sec:simple}

Let us  assume now that the transition kernel is given by a function
\begin{equation}
\Psi_\vecnull(\vecx,\vecv,\xi):=\sigmabar\, \Psi_\vecnull(\vecx,\vecv,\xi,\vecw,\vecz)
\end{equation}
that is independent of $\vecw,\vecz$. As we have seen in Section \ref{sec:d2}, this holds in the two-dimensional setting provided the grain size is less than the mean free path length. 
Then of course also $\Psi(\vecx,\vecv,\xi,\vecw)$ is independent of $\vecw$ and related to the distribution of free path lengths by
\begin{equation}
\Psi(\vecx,\vecv,\xi) = \sigmabar\, \Psi(\vecx,\vecv,\xi,\vecw) .
\end{equation}
Relation \eqref{fitta} becomes \eqref{Cau1}; that is
\begin{equation}\label{three8}
\begin{cases}
\scrD \Psi(\vecx,\vecv,\xi)= \sigmabar\, \Psi_\vecnull(\vecx,\vecv,\xi) & \\
\Psi(\vecx,\vecv,0) = \sigmabar\, \one(\vecx,\vecv) . &
\end{cases}
\end{equation}
The ansatz
\begin{equation}
f_t(\vecx,\vecv,\xi,\vecv_+) = \sigmabar^{\,-1} g_t(\vecx,\vecv,\xi)\, \sigma(\vecv,\vecv_+)
\end{equation}
reduces equation \eqref{glB22} to 
\begin{equation}\label{glB22red}
\wD	 g_t(\vecx,\vecv,\xi) \\
	= \sigmabar^{\,-1}\, \Psi_\vecnull(\vecx,\vecv,\xi) \int_{\S_1^{d-1}}  g_t(\vecx,\vecv_0,0) \,  \sigma(\vecv_0,\vecv)\, d\vecv_0 
\end{equation}
with initial condition 
\begin{equation}
	\lim_{t\to 0}g_t(\vecx,\vecv,\xi) = f_0(\vecx,\vecv)\, \Psi(\vecx,\vecv,\xi) .
\end{equation}
The stationary solution of \eqref{glB22red} corresponding to $f_0(\vecx,\vecv)=1$ is $g_t(\vecx,\vecv,\xi) =\Psi(\vecx,\vecv,\xi)$.

\section{Disordered grains}\label{sec:disorder}

It is instructive to contrast the case of crystal grains discussed above with grains consisting of a disordered medium. We model the medium by scatterers centred at a fixed realisation of a Poisson point process $\scrL_{\text{Poisson}}$ with intensity $1$, rescale by $\epsilon$ and intersect with the grains as in \eqref{PS} to produce the point set
\begin{equation}\label{PS2}
\scrP_\epsilon = \bigcup_i \big( \scrG_i \cap \epsilon \scrL_{\text{Poisson}}\big) .
\end{equation}
If there are no gaps between the grains, $\scrP_\epsilon$ is precisely a fixed realisation of a Poisson process with intensity $1$, for which the convergence to the linear Boltzmann equation has been established in \cite{Boldrighini83}. There do not seem to be any technical obstructions in extending these results to the setting with gaps. In particular, the above results for crystal grains remain valid, if we replace the relevant single-crystal distributions by their disordered counterparts (cf.~\cite{icmp,ICM2014}): For the transition kernels, we have
\begin{equation}
\Phi_\vecnull(\xi,\vecw,\vecz)=\e^{-\sigmabar \xi},\qquad  \Phi(\xi,\vecw)=\e^{-\sigmabar \xi},
\end{equation}
and for the distribution of free path lengths
\begin{equation}
\Phi_\vecnull(\xi,\vecw)=\sigmabar \e^{-\sigmabar \xi}, \qquad
\Phi(\xi)=\sigmabar \e^{-\sigmabar \xi}, \qquad D_\Phi(\xi)=\e^{-\sigmabar \xi}.
\end{equation}
Thus in the case of a disordered granular medium, the formulas 
\eqref{limid}, \eqref{limid2}, \eqref{limid3} and \eqref{limid22} become
\begin{equation}\label{limidgran}
\Psi(\vecx,\vecv,\xi,\vecw)
=\sigmabar^{\,-1}\,\Psi(\vecx,\vecv,\xi)
= \e^{-\sigmabar (\xi-\gap(\vecx,\vecv,\xi))}  \one(\vecx+\xi\vecv,\vecv),
\end{equation}
and
\begin{equation}\label{limid2gran}
\Psi_\vecnull(\vecx,\vecv,\xi,\vecw,\vecz)
= \sigmabar^{\,-1}\, \Psi_\vecnull(\vecx,\vecv,\xi,\vecw)= \e^{-\sigmabar (\xi-\gap(\vecx,\vecv,\xi))}   \one(\vecx,\vecv) \one(\vecx+\xi\vecv,\vecv).
\end{equation} 
The kernel $\Psi_\vecnull(\vecx,\vecv,\xi):=\sigmabar \, \Psi_\vecnull(\vecx,\vecv,\xi,\vecw,\vecz)$ in \eqref{limid2gran} is evidently independent of $\vecw,\vecz$, and thus we are in the setting of Section \ref{sec:simple}. We have 
\begin{equation}
\Psi_\vecnull(\vecx,\vecv,\xi) = \Psi(\vecx,\vecv,\xi) \one(\vecx,\vecv),
\end{equation}
and the ansatz $g_t(\vecx,\vecv,\xi)=f_t(\vecx,\vecv)\Psi(\vecx,\vecv,\xi)$ reduces the generalised Boltzmann equation \eqref{glB22red} to the classical density-dependent linear Boltzmann equation
\begin{equation}\label{glB22red3}
	\big[ \partial_t + \vecv\cdot\nabla_\vecx \big] f_t(\vecx,\vecv) \\
	= \one(\vecx,\vecv) \int_{\S_1^{d-1}}  \big[f_t(\vecx,\vecv_0)-f_t(\vecx,\vecv)\big] \,  \sigma(\vecv_0,\vecv)\, d\vecv_0 
\end{equation}
with initial condition 
\begin{equation}
	\lim_{t\to 0}f_t(\vecx,\vecv) = f_0(\vecx,\vecv) .
\end{equation}

\end{document}